\documentclass[11pt, letterpaper]{amsart}
\usepackage{graphicx, amssymb}
\usepackage{textcomp}

\newtheorem{thm}{Theorem}[section]

\newtheorem{prop}[thm]{Proposition}
\theoremstyle{definition}
\newtheorem{defn}[thm]{Definition}

\theoremstyle{remark}
\newtheorem{rem}[thm]{Remark}
\newtheorem{exa}[thm]{Example}
\numberwithin{equation}{section}

\newcommand{\set}[1]{\left\{#1\right\}}
\newcommand{\Real}{\mathbb R}

\newcommand{\Natural}{\mathbb N}

\newcommand{\such}{{\ | \ }}

\newcommand{\limT}{\lim_{T \to \infty}}
\newcommand{\limsupT}{\limsup_{T \to \infty}}

\newcommand{\dfn}{\, := \,}

\newcommand{\prob}{\mathbb{P}}

\newcommand{\plim}{\prob \text{-} \lim}
\newcommand{\plimsup}{\prob \text{-} \limsup}
\newcommand{\pliminf}{\prob \text{-} \liminf}

\newcommand{\plimT}{\plim_{T \to \infty}}
\newcommand{\plimsupT}{\plimsup_{T \to \infty}}
\newcommand{\pliminfT}{\pliminf_{T \to \infty}}

\newcommand{\qprob}{\mathbb{Q}}
\newcommand{\expec}{\mathbb{E}}

\newcommand{\basis}{(\Omega, \mathcal{F}, \filtration, \prob)}
\newcommand{\filtration}{\pare{\mathcal{F}_t}_{t \in \Real_+}}

\newcommand{\F}{\mathcal{F}}
\newcommand{\G}{\mathcal{G}}
\newcommand{\Hcal}{\mathcal{H}}
\newcommand{\cadlag}{c\`adl\`ag}

\newcommand{\ud}{\, \mathrm d}

\newcommand{\bbe}{\Phi^\downarrow}
\newcommand{\bab}{\Phi^\uparrow}
\newcommand{\bou}{\Phi}

\newcommand{\subseteqp}{\subseteq_\prob}
\newcommand{\eqp}{=_\prob}

\newcommand{\pare}[1]{\left(#1\right)}
\newcommand{\bra}[1]{\left[#1\right]}
\newcommand{\dbra}[1]{[\kern-0.15em[ #1 ]\kern-0.15em]}
\newcommand{\dbraco}[1]{[\kern-0.15em[ #1 [\kern-0.15em[}

\newcommand{\Op}{O_\prob}
\newcommand{\Opp}{O^\uparrow_\prob}
\newcommand{\Opn}{O^\downarrow_\prob}

\newcommand{\indic}{\mathbb{I}}

\newcommand{\absco}{{<\kern-0.53em<}}

\begin{document}

\title[On the Dybvig-Ingersoll-Ross Theorem]{On the Dybvig-Ingersoll-Ross Theorem}%
\author{Constantinos Kardaras}%
\address{Constantinos Kardaras, Mathematics and Statistics Department, Boston University, 111 Cummington Street, Boston, MA 02215, USA.}%
\email{kardaras@bu.edu}%
\author{Eckhard Platen}%
\address{Eckhard Platen, School of Finance and Economics \& Department of Mathematical Sciences, University of Technology, Sydney, P.O. Box 123, Broadway, NSW 2007, Australia.}%
\email{eckhard.platen@uts.edu.au}%

\thanks{The first author would like to acknowledge the generous support of the \emph{Bruti-Liberati Visiting Fellowship} that enabled his visit at the \emph{School of Finance and Economics} of the \emph{University of Technology, Sydney}, where this work was carried out.}%

\date{\today}%
\begin{abstract}
The Dybvig-Ingersoll-Ross (DIR) theorem states that, in arbitrage-free term structure models, long-term yields and forward rates can never fall. We present a refined version of the DIR theorem, where we identify the reciprocal of the maturity date as the \emph{maximal} order that long-term rates at earlier dates can dominate long-term rates at later dates. The viability assumption imposed on the market model is weaker than those appearing previously in the literature.
\end{abstract}

\maketitle

\setcounter{section}{0}

\section{Introduction}

\subsection{Background and discussion of the results}

In interest-rate modeling, it is a well-known result that \emph{if the market is arbitrage-free, then long-maturity  yields, as well as forward rates, can never fall}. The last statement is commonly referred to as the \textsl{Dybvig-Ingersoll-Ross} (DIR) theorem, acknowledging the fact that its first occurrence was in \cite{DIR} and opened this research direction. Since then, there has been substantial interest in the literature regarding this result: \cite{RePEc:osu:osuewp:00-12} contained some clarifications on the original proof. Later, \cite{MR1926241} presented an elegant mathematical proof in a quite general context.
Recently, \cite{GolSch08} discussed further interesting generalizations, as well as an asymptotic minimality property, also appearing in \cite{schulze}.

In order to get a better feeling for what the DIR theorem states, let $P_t^T$ denote the price at time $t \in \Real_+$ of a zero-coupon bond with maturity $T > t$; then,
\begin{equation} \label{eq: yield}
R_t^T = - \frac{\log(P_t^T)}{T - t}
\end{equation}
is the prevailing yield from time $t$ to maturity $T$. By ``long-maturity yield at time $t$'', one usually means the \emph{limit} of $R_t^T$ as $T \to \infty$, which, provided it can be defined in some sense, we denote by $R^\infty_t$. The DIR theorem states that, under the assumption of absence of arbitrages in the market, $R^\infty_s \leq R^\infty_t$ holds whenever $s \leq t$. A completely similar statement is valid for forward rates; to refrain ourselves from being repetitive, we shall focus on yields for the purposes of the introductory discussion here.

Originally, the DIR theorem is stated for term-structure models of interest rates. We choose here to take the more comprehensive viewpoint of the term-structure of a market for exchange over time of some underlying asset, which could be a currency, a commodity with investment value, or a similar security. Within this framework, $P_t^T$ represents the units of the underlying asset required by the market at time $t \in \Real_+$ in return of one unit of the underlying asset at time $T > t$. In other words, $P_t^T$ denotes the price, in units of the asset, of a derivative contract that allows transferring the asset through time; as such, it is therefore deeply linked to the term structure of yields and forward rates.

\smallskip

Having clarified the background and statement of the DIR theorem in this general context, two natural questions come to mind:
\begin{enumerate}
  \item What can we salvage if $R^\infty_t$ \emph{cannot} be defined for some $t \in \Real_+$, i.e., if limits of yields as the maturity tends to infinity do \emph{not} exist?
  \item For long-term, but \emph{finite} maturities $T$, the relation $R_s^T \leq R_t^T$, for $s \leq t$, \emph{might} fail to hold. How large can the discrepancy $R_s^T - R_t^T$ be?
\end{enumerate}

An approach to answering the first question is undertaken in \cite{GolSch08}. There, an appropriate \emph{superior limit} definition is utilized in order to compensate for the possible nonexistence of the actual limit. In fact, the authors give a reasonable economic justification for considering the aforementioned superior limit. The approach we take here is to consider the difference $R_s^T - R_t^T$ for $s \leq t$ as $T \to \infty$, and examine when its superior limit (in probability) exists and is nonnegative. Though the previous two approaches are similar in nature, focusing on the \emph{difference} of the rates allows for more detailed comparisons. An example of such instance would be the case where long-term rates explode in the limit.

To the best of our knowledge, an attempt to answer the second question posed above has not appeared in the literature. We show here that the highest possible order that $R_s^T$ can be larger than $R_t^T$ is $1 / T$, i.e., the reciprocal of the long-term maturity. In fact, we shall show by example that this order is the best possible that can be achieved.

\smallskip

As mentioned earlier, and as easy counterexamples show, the DIR theorem is valid only under an assumption regarding nonexistence of some sort of arbitrages in the market. In the literature, there had been mainly two approaches in formalizing such an assumption:
\begin{itemize}
  \item In the first approach, authors stipulate a ``no limiting arbitrage'' condition in the market, reminiscent of the ``No Free Lunch with Vanishing Risk'' condition introduced in \cite{MR1304434}. This was for example the approach initially taken in \cite{DIR}, as well as in \cite{RePEc:osu:osuewp:00-12} shortly after. More recently, \cite{schulze} also takes the same path.
   \item The second approach is to assume the existence of a locally equivalent martingale measure (EMM) in the market. Bond prices are defined as expectations under the EMM of contingent claims giving unit payoff at maturity, discounted by the savings account. This viewpoint on the statement of the DIR theorem was initiated in \cite{MR1926241}.
\end{itemize}
The Fundamental Theorem of Asset Pricing, established in \cite{MR1304434} for the case of equity markets, indicates that the previous assumptions are very closely connected. However, the fact that a continuum of assets is available to trade in bond markets forces different tools to be employed under the two approaches above. This is true even in papers who treat both cases, like \cite{GolSch08}.

Here, we take a path that unifies the above two approaches, at the same time weakening market viability assumptions that have previously appeared. This is done by assuming existence of \textsl{strictly positive supermartingale deflators} in the market, an assumption weaker than the existence of an EMM, and equivalent to absence of arbitrages of the first kind in the bond market where only long positions are allowed, as is discussed in \cite{Kar08}.

\bigskip

After a few probabilistic definitions and later needed results will conclude this section, the structure of the remaining paper is as follows: In Section \ref{sec: results} all the results are presented, while Section \ref{sec: exas} contains examples that illustrate our main findings.

\subsection{Probabilistic definitions and notation}

Let $(\Omega, \, \F, \, \prob)$ be a probability space where all the random elements appearing below will be based.

For $A \in \F$ and $B \in \F$, we write $A \subseteqp B$  if and only if $\prob \bra{(\Omega \setminus B) \cap A} = 0$ --- in other words, $A \subseteqp B$ means that $A$ is contained in $B$ modulo $\prob$. Also, $A \eqp B$ means $A$ and $B$ are equal modulo $\prob$, i.e.,  that both $A \subseteqp B$ and $B \subseteqp A$ hold.

For a collection $(\xi^T)_{T \in \Real_+}$ of random variables, $\plimsupT \xi^T$ is defined to be the essential infimum of all random variables $\zeta$ such that $\lim_{T \to \infty} \prob[\xi^T \leq \zeta] = 1$. Observe that $\plimsupT \xi^T$ is an \emph{extended-valued} random variable, i.e., it can potentially take infinite values, both positive and negative. We also define $\pliminfT \xi^T \dfn - \plimsupT (- \xi^T)$. The limit in probability of $(\xi^T)_{T \in \Real_+}$ as $T \to \infty$ exists if and only if $\pliminfT \xi^T = \plimsupT \xi^T$; in this case, this limit is denoted by $\plimT \xi^T$. (For these definitions and more discussion, we refer the reader to Chapter I of \cite{MR1219534}.)

Let again $(\xi^T)_{T \in \Real_+}$ be a collection of random variables. Whenever
\[
\lim_{\ell \to \infty} \pare{\limsup_{T \to \infty} \prob \bra{\xi^T > \ell}} = 0
\]
holds, we shall be writing $\xi^T = \Opp (1)$ as $T \to \infty$. Also, if $(\alpha^T)_{T \in \Real_+}$ is a sequence of strictly positive real numbers and $A \in \F$, we write $\xi^T = \Opp (\alpha^T)$ on $A$ as $T \to \infty$ if and only if $\indic_A \xi^T / \alpha^T = \Opp (1)$ as $T \to \infty$, where $\indic_A$ denotes the \textsl{indicator} function of the event $A$. Furthermore, we write $\xi^T = \Opn (\alpha^T)$ on $A$ as $T \to \infty$ if and only if $- \xi^T = \Opp (\alpha^T)$ on $A$ as $T \to \infty$. Finally, $\xi^T =  \Op (\alpha^T)$ on $A$ as $T \to \infty$ means $|\xi^T| = \Opp (\alpha^T)$ on $A$ as $T \to \infty$. If the set $A \in \F$ is not explicitly mentioned, it will be tacitly assumed that $A = \Omega$.

As the reader might have already guessed, we are using throughout the upwards-pointing arrow ``$\uparrow$'' as a mnemonic device when dealing with boundedness from above; similarly, the downwards-pointing arrow ``$\downarrow$'' will be used in cases where boundedness from below is involved.

\smallskip

We close this introductory discussion with two general results, which will be used in the text.
\begin{prop} \label{prop: plimsup stuff}
In the statements below, $(\xi^T)_{T \in \Real_+}$ and $(\zeta^T)_{T \in \Real_+}$ are collections of random variables and $(\alpha^T)_{T \in \Real_+}$ is a collection of strictly positive real numbers.
\begin{enumerate}
  \item $\plimsupT \xi^T < \infty$ implies that $\xi^T = \Opp (1)$ as $T \to \infty$.
  \item If $\xi^T = \Opp (\alpha^T)$ as $T \to \infty$ and $\lim_{T \to \infty} \alpha^T = 0$, then $\plimsupT \xi^T \leq 0$.
  \item If $\plimsupT(\xi^T - \zeta^T) \leq 0$, then $\plimsupT \xi^T \, \leq \, \plimsupT\zeta^T$.
\end{enumerate}
\end{prop}

\begin{proof}
\noindent (1) Let $\overline{\xi} = \plimsupT \xi^T$. Fix $\ell \in \Real_+$. The set-inclusion $\set{\overline{\xi} \leq \ell - 1} \bigcap \set{\xi^T \leq \overline{\xi} + 1} \subseteq  \set{\xi^T \leq \ell}$, valid for all $T \in \Real_+$, gives
\begin{equation} \label{eq: help for limsup stuff}
\prob \bra{\xi^T > \ell} \leq \prob \bra{\overline{\xi} > \ell - 1} + \prob \bra{\xi^T > \overline{\xi} + 1}.
\end{equation}
As $\prob \bra{\overline{\xi} < \infty} = 1$, we get $\limT \prob \bra{\xi^T > \overline{\xi} + 1} = 0$ from the definition of $\plimsupT \xi^T$. Therefore, \eqref{eq: help for limsup stuff} gives $\limsupT \prob \bra{\xi^T > \ell} \leq \prob \bra{\overline{\xi} > \ell - 1}$. Using again $\prob \bra{\overline{\xi} < \infty} = 1$ we get $\lim_{\ell \to \infty} \prob \bra{\overline{\xi} > \ell - 1} = 0$; therefore, $\lim_{\ell \to \infty} \pare{\limsupT \prob \bra{\xi^T > \ell}} = 0$, which is what we needed to prove.

\smallskip
\noindent (2) Let $\epsilon > 0$. Then,
\[
\limsup_{T \to \infty} \prob[\xi^T > \epsilon] = \limsup_{T \to \infty} \prob[\xi^T / \alpha^T > \epsilon / \alpha^T] \leq \limsup_{T \to \infty} \prob[\xi^T / \alpha^T > \ell]
\]
holds for all $\ell > 0$ in view of $\lim_{T \to \infty} \alpha^T = 0$. Taking limits as $\ell \to \infty$ in the extreme sides of the previous inequality we obtain $\limsup_{T \to \infty} \prob[\xi^T > \epsilon] = 0$, which means that $\plimsupT \xi^T \leq \epsilon$. As this holds for all $\epsilon > 0$, we get $\plimsupT \xi^T \leq 0$.

\smallskip

\noindent (3) Take any random variable $\eta$ such that $\limT [\zeta^T \leq \eta] = 1$. For any $\epsilon > 0$, we have
\[
\limsupT \prob[\xi^T > \epsilon + \eta] \leq \limsupT \prob[\zeta^T > \eta ] + \limsupT \prob[\xi^T - \zeta^T > \epsilon] = 0.
\]
This implies that $\plimsupT \xi^T \leq \epsilon + \plimsupT \zeta^T$ for all $\epsilon > 0$. Letting now $\epsilon$ tend to zero, we get the result.
\end{proof}

\begin{prop} \label{prop: yields bdd}
Let $(\xi^T)_{T \in \Real_+}$ be a collection of random variables. Then, the following statements are true:
\begin{enumerate}
     \item There exists $\bbe \in \F$ such that: $\xi^T =  \Opn(1)$ on $A \in \F$ as $T \to \infty$ \emph{if and only if} $A \subseteqp \bbe$.
     \item There exists $\bab \in \F$ such that: $\xi^T =  \Opp(1)$ on $A \in \F$ as $T \to \infty$ \emph{if and only if} $A \subseteqp \bab$.
     \item There exists $\bou \in \F$ such that: $\xi^T =  \Op(1)$ on $A \in \F$ as $T \to \infty$ \emph{if and only if} $A \subseteqp \bou$.
\end{enumerate}
Furthermore, the sets $\bbe$, $\bab$ and $\bou$ are unique modulo $\prob$.
\end{prop}

\begin{proof}
We only prove statement (1); the proofs of statement (2) and statement (3) are entirely similar.

Consider the class $\G^\downarrow \dfn \big\{ A \in \F \such \xi^T = \Opn(1) \text{ holds on } A \text{ as } T \to \infty \big\} \subseteq \F$. Since  $\emptyset \in \G^\downarrow$, the class $\G^\downarrow$ is nonempty. Furthermore, it is relatively straightforward to see that $\G^\downarrow$ is closed under countable unions. Observe that $\subseteqp$ is a partial ordering on the subsets of $\F$. Let $\Hcal \subseteq \G^\downarrow$ be a totally ordered set for the order $\subseteqp$ and let $p \dfn \sup \set{\prob[A] \such A \in \Hcal}$. For all $n \in \Natural$, pick $A^n \in \Hcal$ such that $\prob[A^n] \geq p - 1/n$. If $A \dfn \bigcup_{n \in \Natural} A^n$, then $A \in \G^\downarrow$ and it is straightforward that $A$ is an upper bound of $\Hcal$. Zorn's lemma then implies the existence of a \emph{maximal} element in $\G^\downarrow$. Since $\G^\downarrow$ is closed with respect to finite unions, we conclude that the previous maximal element is unique, which we call $\bbe$. The uniqueness modulo $\prob$ of such set $\bbe$ follows immediately from statement (1) of the result.
\end{proof}

\section{Results} \label{sec: results}

\subsection{Market model and yields}
On a filtered probability space $\basis$, we consider a collection $\pare{P^T}_{T \in \Real_+}$ of \cadlag \ (right continuous with left-hand limits) stochastic processes indexed by their \textsl{maturity} $T \in \Real_+$. For each $T \in \Real_+$, $P^T$ is defined in the finite time interval $[0,T]$, i.e., $P^T = (P^T_t)_{t \in [0, T]}$. We assume that $\prob[P_t^T > 0] = 1$ holds for all $t \in [0, T]$ and $T \in \Real_+$, as well as $\prob[P_T^T = 1] = 1$. For a concrete interpretation, regard $P_t^T$ as the price at time $t$ of an instrument delivering a unit of account at time $T \geq t$. Observe however that we do not necessarily assume that $P^T \leq 1$, which is true in bond markets. This is done for a number of reasons:
\begin{enumerate}
  \item From a theoretical viewpoint, $P^T \leq 1$ is not needed for the results we shall present.
  \item From a model-building perspective, such assumption would immediately disqualify all Gaussian short-rate models that are widely used in the industry.
  \item On a more practical side, and as mentioned in the Introduction, our results are applicable in diverse situations, such as commodity markets. If the storage costs that apply for the commodity involved, which could be for example oil, are more than the convenience yield it carries, it is certainly possible that $P_t^T > 1$ holds for $t < T$.
\end{enumerate}

For $0 \leq t < T$, the \emph{yield $R_t^T$ from time $t$ to maturity $T$} is defined in \eqref{eq: yield}. Events where long-term yields are essentially bounded will turn out to be crucial in our discussion. In all that follows, for $t \in \Real_+$, we use  $\bbe_{t}$, $\bab_t$ and $\bou_t$ to be the events appearing in the statement of Proposition \ref{prop: yields bdd} corresponding to the case where $\xi^T = R_t^T$ for $T > t$.
It is apparent that $\bbe_t$ is the \emph{maximal} (modulo $\prob$) event such that long term yields at time $t$ are bounded in probability from below. Exactly similar characterizations are true for $\bab_t$ and $\bou_t$ in terms of boundedness in probability from above and two-sided boundedness in probability, respectively. Obviously, $\bou_t \eqp \bbe_t \cap \bab_t$ holds for all $t \in \Real_+$.

\begin{rem}
In bond markets, we have $P^T \leq 1$ for all $T \in \Real_+$, or equivalently that $R^T \geq 0$
for all $T \in \Real_+$. Therefore, $\bbe_t \eqp \Omega$ for all $t \in \Real_+$; in other words, long-term yields trivially are essentially bounded from below at every time $t \in \Real_+$.
\end{rem}

\begin{rem}
It has been empirically observed that yield curves flatten out for very long maturities; a discussion on this appears for example in \cite{Malkiel}. There also exist theoretical justifications of this phenomenon, as is described in \cite{KFG} and \cite{Yao1999327}. To rigorously describe such behavior in a  weak sense, assume that $\plimT R_t^T$ exists and is a $\prob$-a.s. finite random variable for a fixed $t \in \Real_+$. Then, statement (1) of Proposition \ref{prop: plimsup stuff} implies that $\bou_t \eqp \Omega$.
\end{rem}

\begin{rem}
In this paper, we treat continuous-time models --- for this reason, we use the definition \eqref{eq: yield} for yields. We note, however, that all our results still hold in discrete-time (infinite horizon) settings, with the appropriate changes in the definition of yields and forward rates (see, for example, equations (2.1) and (2.3) of \cite{GolSch08}). The details have been extensively discussed in \cite{MR1926241} and \cite{GolSch08}, where we refer the interested reader.
\end{rem}

\subsection{Strictly positive supermartingale deflators} The notion introduced below is central in our discussion.

\begin{defn} \label{defn: supermat defl}
A \textsl{strictly positive supermartingale deflator} in the market is a \cadlag \ process $Y$ with $\inf_{t \in [0, T]} Y_t > 0$, $\prob$-a.s., for all $T \in \Real_+$, such that $(Y_t P_t^T)_{t \in [0, T]}$ is a supermartingale for all $T \in \Real_+$.
\end{defn}

Existence of a strictly positive supermartingale deflator is equivalent to absence of arbitrages of the first kind in the market with acting investors that may only take long positions on the instruments with prices $(P^T)_{T \in \Real_+}$. For such ``abstract'' markets with infinite number of assets, the last fact is explained in detail in \cite{Kar08}.

\begin{rem}
Even if  the processes $(P^T)_{t \in [0, T]}$ for $T \in \Real_+$ are not initially assumed to have \cadlag \ paths, but are only right-continuous in probability, the existence of a strictly positive supermartingale deflator, as in Definition \ref{defn: supermat defl}, coupled with the standard supermartingale modification theorem, implies that there exist \cadlag \ modifications of $(P^T)_{t \in [0, T]}$, $T \in \Real_+$. As every model encountered in practice consists of \cadlag \  price-processes, we plainly enforce this requirement from the outset.
\end{rem}

\smallskip

We shall now discuss the traditional way of constructing markets possessing a strictly positive supermartingale deflator, via the existence of an equivalent martingale measure (EMM). We include this discussion for completeness since we shall be using it in the examples below. It is important to note that markets where a strictly positive supermartingale deflator exists form a wide-encompassing class, substantially larger than the concrete situation described in the example below. A concrete realistic example where an EMM fails to exist, but a strictly positive supermartingale deflator \emph{does} exist, is presented in \S3.2 of \cite{RePEc:uts:rpaper:198}; in this respect, see also \S \ref{exa: determin counterexa} of the present paper.

\begin{exa} \label{exa: EMM construction}
Let $\qprob$ be a probability on $(\Omega, \F)$ such that $\qprob$ is equivalent to $\prob$ on $\F_t$ for all $t \in \Real_+$. Consider also a \cadlag \ nonnegative process  $B$, representing the \emph{savings account}, such that $\prob \bra{\inf_{t \in [0, T]} B_t > 0} = 1$ as well as $\expec^\qprob[1 / B_T] < \infty$, for all $T \in \Real_+$. Define $P^T$ to be the \cadlag \ modification of the process $[0, T] \ni t \mapsto B_t \, \expec^\qprob \bra{1 / B_T \such \F_t}$. For this market, a strictly positive supermartingale deflator exists and is given by
\[
Y \dfn \frac{1}{B} \, \frac{\ud (\qprob|_{\F_\cdot})}{\ud (\prob|_{\F_\cdot})}.
\]
(In fact, one should consider the \cadlag \ version of the process above.) Indeed, it is straightforward to check that $(Y_t P_t^T)_{t \in [0, T]}$ is actually a $\prob$-\emph{martingale} for all $T \in \Real_+$.
\end{exa}

Contrary to the construction in Example \ref{exa: EMM construction} above, we do not explicitly define a savings account here, as it is not needed. At any rate, \emph{given} a market with prices $(P^T)_{T \in \Real_+}$, \emph{if} a savings account $B$ is able to generate the market in the sense of Example \ref{exa: EMM construction}, i.e., if $P_t^T = B_t \, \expec^\qprob \bra{1 / B_T \such \F_t}$ holds for all $t \leq T$ where $\qprob$ is equivalent to $\prob$ on $\F_t$ for all $t \in \Real_+$, then $B$ is essentially unique; see \cite{MR1779587}.

\subsection{Long-term yields}

We are ready to state the main result of the paper, which can be regarded as a ramification of the DIR theorem.

\begin{thm} \label{thm: DIR for yields}
Suppose that a strictly positive supermartingale deflator exists in the market. Let $s \leq t$. Then:
\begin{enumerate}
  \item $\bbe_s \subseteqp \bbe_t$.
  \item $R^T_{s} - R^T_{t} =  \Opp(1 / T)$ holds on $\bbe_t$  as $T \to \infty$.
\end{enumerate}
\end{thm}

\begin{proof}
For all $T \in \Real_+$, define $L^T = (L_t^T)_{t \in [0, T]}$ via $L^T \dfn Y P^T$, where $Y$ is a strictly positive supermartingale deflator as in Definition \ref{defn: supermat defl}. Then, $(L_t^T)_{t \in [0, T]}$ is a nonnegative supermartingale and $(T- u) R_u^T = - \log (L^T_u) + \log (Y_u)$ holds whenever $u < T$. Write
\begin{equation} \label{eq: good decompo}
(T - t) \pare{R^T_s - R^T_t} = - (t - s) R_s^T + \log \pare{\frac{L_t^T}{L_s^T}} - \log \pare{\frac{Y_t}{Y_s}}
\end{equation}
Let $\ell > 0$; then, we have
\[
\prob \big[ \log (L_t^T / L_s^T) >  \ell  \big] = \prob \big[ L_t^T / L_s^T > e^{\ell} \big] \leq e^{-\ell},
\]
following from Markov's inequality, since $L^T$ is a nonnegative supermartingale. This implies that $\log (L_t^T / L_s^T) = \Opp(1)$ as $T \to \infty$. Since $\log(Y_t / Y_s)$ is an $\Real$-valued random variable and $- (t-s) R^T_s = \Opp(1)$ holds on $\bbe_s$ as $T \to \infty$, \eqref{eq: good decompo} gives that $(T- t) \pare{R^T_{s} - R^T_{t}} =  \Opp(1)$ on $\bbe_s$ as $T \to \infty$. As this obviously implies that $R^T_{s} - R^T_{t} =  \Opp(1)$ on $\bbe_s$ as $T \to \infty$, we obtain that
\[
T \pare{R^T_{s} - R^T_{t}} =  (T- t) \pare{R^T_{s} - R^T_{t}} + t \pare{R^T_{s} - R^T_{t}} = \Opp(1) \text{ holds on } \bbe_s \text{ as } T \to \infty,
\]
which is the same as saying that $R^T_{s} - R^T_{t} = \Opp(1/T)$ on $\bbe_s$ as $T \to \infty$. This immediately implies that $\bbe_s \subseteqp \bbe_t$.

Up to now we have proved that $R^T_{s} - R^T_{t} = \Opp(1/T)$ on $\bbe_s$ as $T \to \infty$; we would like to extend the last relationship to hold on $\bbe_t$. Provided that we replace \eqref{eq: good decompo} with
\[
(T - s) \pare{R^T_s - R^T_t} =  - (t - s) R_t^T + \log \pare{\frac{L_t^T}{L_s^T}} - \log \pare{\frac{Y_t}{Y_s}},
\]
one can follow essentially the same steps as above to finish the proof.
\end{proof}

\subsection{The DIR theorem revisited}

Let $s \leq t$. Theorem \ref{thm: DIR for yields} coupled with statement (2) of Proposition \ref{prop: plimsup stuff} immediately gives that $\plimsupT \pare{R^T_{s} - R^T_{t}} \leq 0$ holds on $\bbe_t$. In particular, and using statement (3) of Proposition \ref{prop: plimsup stuff}, we obtain that
\begin{equation} \label{eq: DIR classical}
\plimsupT R^T_{s} \, \leq \, \plimsupT R^T_{t} \ \text{ holds on } \bbe_t.
\end{equation}

The last equation \eqref{eq: DIR classical} should be compared with the result obtained in \cite{GolSch08}. Of course, in \cite{GolSch08} the superior limit is taken in a stronger sense and the assumption that we are working on $\bbe_t$ is not present. It is indeed true that \eqref{eq: DIR classical} can be still valid outside of $\bbe_t$, \emph{even though} $\plimsupT \pare{R^T_{s} - R^T_{t}} > 0$. Such a situation is described in \S \ref{subsubsec: DIR classical gives no info}; there, both sides of \eqref{eq: DIR classical} are equal to infinity, and are, therefore, equal in a trivial sense. Theorem \ref{thm: DIR for yields} refines the asymptotic relationship \eqref{eq: DIR classical} by precisely examining the behavior of the \emph{relative} differences of long-term yields through different points in time.
%
%
%

\subsection{Forward rates}

The next aim is to obtain an equivalent of Theorem \ref{thm: DIR for yields} for forward rates, which we now introduce. For $0 < t < t' \leq T$, the \textsl{forward rate, set at time $t$ for investment from time $t'$ up to maturity $T$}, is defined via
\begin{equation} \label{eq: forward rate}
F^T_{t, t'} \dfn  \frac{1}{T - t'} \, \log \pare{\frac{P^{t'}_{t}}{P^T_t}} \, = \, \frac{T - t}{T - t'} R^T_t - \frac{t' - t}{T - t'} R^{t'}_t.
\end{equation}

\smallskip

Roughly speaking, the next result we shall present states that yields are essentially bounded exactly on the set where forward rates and yields are asymptotically, as $T \to \infty$, equivalent of order $1 / T$. Similar statements hold for boundedness from below and above. Observe that there is \emph{no} market viability assumption in the statement of Proposition \ref{prop: forwards equiv to yields}.

\begin{prop} \label{prop: forwards equiv to yields}
Let $t \in \Real_+$ and $A \in \F$. The following conditions are equivalent:
\begin{enumerate}
  \item $R^T_{t} =  \Opn(1)$ on $A$ as $T \to \infty$.
  \item For \emph{all} $t' > t$, $F_{t, t'}^T - R_t^T = \Opn(1 / T)$ holds on $A$ as $T \to \infty$.
  \item For \emph{some} $t' > t$, $F_{t, t'}^T - R_t^T = \Opn(1 / T)$ holds on $A$ as $T \to \infty$.
\end{enumerate}
The same equivalences hold if we replace ``$\Opn$'' with ``$\Opp$'' in all conditions (1), (2) and (3), and similarly if we replace ``$\Opn$'' with ``$\Op$'' in all conditions (1), (2) and (3).
\end{prop}

\begin{proof}
We shall only prove the equivalence of (1), (2) and (3) as explicitly stated in Proposition \ref{prop: forwards equiv to yields}. The cases where we replace ``$\Opp$'' with ``$\Opn$'' or ``$\Op$'' in all conditions (1), (2) and (3) is entirely similar.
In what follows, $t \in \Real_+$ and $A \in \F$ are fixed.

Start by assuming (1) and fix $t' > t$. First of all, observe that $F_{t, t'}^T = \Opn(1)$ on $A$ as $T \to \infty$, as follows from the fact that $R_{t}^T = \Opn(1)$ on $A$ as $T \to \infty$ and the definition of the forward rates at \eqref{eq: forward rate}. Now, $(T - t') \big( F_{t, t'}^T - R_t^T \big) = (t' - t) \big( R_t^T - R_t^{t'} \big)$ as follows again from \eqref{eq: forward rate}, immediately gives that $(T - t') (F_{t, t'}^T - R_t^T ) = \Opn(1)$ on $A$ as $T \to \infty$, since $R_t^T = \Opn(1)$ on $A$ as $T \to \infty$. Using also the fact that $F_{t, t'}^T - R_t^T  = \Opn(1)$ on $A$  as $T \to \infty$, we get that $T (F_{t, t'}^T - R_t^T) = \Opn(1)$ on $A$ as $T \to \infty$, which is what we needed to show.

Of course, condition (2) implies condition (3).

Now, assume (3). Observe first that
\[
(T - t') \pare{F_{t, t'}^T - R_t^T } = \pare{\frac{T - t'}{T}} \, T  \pare{F_{t, t'}^T - R_t^T } = \Opn(1) \text{ holds on $A$ as } T \to \infty.
\]
Then,
\[
R_t^T =  \pare{\frac{T - t'}{t' - t}} \pare{F_{t, t'}^T - R_t^T }  + R_t^{t'} = \Opn(1) \text{ holds on $A$ as }  T \to \infty,
\]
which is exactly condition (1) and concludes the proof.
\end{proof}

According to Proposition \ref{prop: forwards equiv to yields} and Proposition \ref{prop: yields bdd}, $\bbe_t$ can be regarded as the largest set where $F_{t, t'}^T - R_t^T = \Opn(1 / T)$ holds for some, and then for all, $t' > t$. Similar interpretations are valid for the events $\bab_t$ and $\bou_t$, where $t \in \Real_+$.

\smallskip

We are now ready to state the version of Theorem \ref{thm: DIR for yields} for forward rates. The situation is only slightly more complicated, since we have to control the boundedness of yields from both sides at different points in time.

\begin{thm} \label{thm: DIR for forward}
Suppose that a strictly positive supermartingale deflator exists in the market. Let $s \leq t$, as well as $s < s'$ and $t < t'$. Then, $F_{s, s'}^T - F_{t, t'}^T = \Opp(1 / T)$ holds on $\bab_s \cap \bbe_t$ as $T \to \infty$.
\end{thm}

\begin{proof}
Write
\[
F_{s, s'}^T - F_{t, t'}^T = \big( F_{s, s'}^T - R_{s}^T    \big) - \big( F_{t, t'}^T -  R_{t}^T \big) + \big( R_{s}^T -  R_{t}^T \big).
\]
Now, $F_{s, s'}^T - R_{s}^T = \Opp(1 / T)$ and $R_{t} - F_{t, t'}^T = \Opp(1 / T)$ and  both hold on $\bab_s \cap \bbe_t$ as $T \to \infty$ in view of Proposition \ref{prop: forwards equiv to yields}. Furthermore, $R_{s}^T -  R_{t}^T = \Opp(1 / T)$ holds on $\bbe_t$ by Theorem \ref{thm: DIR for yields}. Putting everything
together, we obtain the claim of Theorem \ref{thm: DIR for forward}. \end{proof}

\begin{rem}
Let $s \leq t$. If a strictly positive supermartingale deflator exists in the market, statement (1) of Theorem \ref{thm: DIR for yields} gives $\bou_s \eqp \bab_s \cap \bbe_s \subseteqp \bab_s \cap \bbe_t$. In particular, Theorem \ref{thm: DIR for forward} implies that $F_{s, s'}^T - F_{t, t'}^T = \Opp(1 / T)$ holds on $\bou_s$ as $T \to \infty$, whenever $s < s'$ and $t < t'$, which is a more pleasant statement.
\end{rem}

\section{Remarks and Examples} \label{sec: exas}

We proceed with several remarks and (counter)examples regarding our main results. The most important ones are given in \S \ref{exa: on best rate}, where it is shown that the reciprocal of the maturity is indeed the best order of domination that can be obtained, and \S \ref{exa: determin counterexa}, where we demonstrate that our market viability assumption is strictly weaker than the ones that previously appeared in the literature.

\subsection{Counterexamples on the main results}

\subsubsection{}
The inclusion $\bbe_s \subseteqp \bbe_t$ in Theorem \ref{thm: DIR for yields} might fail when a strictly positive supermartingale deflator does not exist. Consider for example the deterministic market with $P_t^T = 1$ for $0 \leq t < 1$ and $t \leq T$, while $P_t^T = \exp(T^2 - t^2)$ for $1 \leq t \leq T$. Then, $R_0^T = 0$ and $R_1^T = - T - 1$ holds for $T \geq 1$. Therefore,  $\bbe_0 \eqp \Omega \supsetneq_\prob \emptyset \eqp \bbe_1$.

\subsubsection{} \label{subsubsec: DIR classical gives no info}
Even when a strictly positive supermartingale deflator exists, the asymptotic behavior of yield differences mentioned in statement (2) of Theorem \ref{thm: DIR for yields} can fail to hold outside $\bbe_t$. With $\qprob = \prob$ and $B$ defined via $B_t = \exp(-t^2)$ for $t \in \Real_+$, define a market according to Example \ref{exa: EMM construction}. In this case, $\bbe_t \eqp \emptyset$ for all $t \in \Real_+$. Further, $R_t^T = - T - t$ for $t \leq T$, which implies that $R_s^T - R_t^T = t - s > 0$ for $s < t$, and statement (2) of Theorem \ref{thm: DIR for yields} fails to hold. Observe also in this example that the asymptotic relationship $\limsup_{T \to \infty} R^T_{s} = \infty = \limsup_{T \to \infty} R^T_{t}$ trivially holds identically; however, one cannot honestly claim that long-term yields are nonincreasing, as $\lim_{T \to \infty} (R_s^T - R_t^T) = t - s > 0$ whenever  $s < t$.

\subsubsection{} \label{rem: nice strange bond market}

With $\qprob = \prob$ and $B$ defined via $B_t = \exp(t^2)$ for $t \in \Real_+$, define a bond market according to Example \ref{exa: EMM construction}. By construction, there exists a strictly positive supermartingale deflator. Furthermore, $\bbe_t \eqp \Omega$ holds for all $t \in \Real_+$, and we have $R_t^T = T + t$ for $t \leq T$.

In the setting of Theorem \ref{thm: DIR for yields}, this example shows that $\plimT \pare{R_s^T - R_t^T}$ exists and is \emph{strictly} negative on $\bbe_t$ for $s < t$. Indeed, this follows by observing that $R_s^T - R_t^T = -(t - s) < 0$ for $s < t$.

We move on to the setting of Theorem \ref{thm: DIR for forward}. A straightforward use of \eqref{eq: forward rate} gives that, for $0 \leq t < t' < T$, $F_{t, t'}^T = T + t'$. Pick $s \leq t$, $s < s'$, $t < t'$; then, $\limT (F_{s, s'}^T - F_{t, t'}^T) = s' - t'$, which can take \emph{any} value in $\Real$ for appropriate choices of $s'$ and $t'$. Therefore, this example shows that we can have $\plimT \big( F_{s, s'}^T - F_{t, t'}^T  \big) < 0$ on $\bbe_t$, if we are not working on $\bab_s$, which shows the sharpness of the result in Theorem \ref{thm: DIR for forward}.

\subsection{Optimal rate} \label{exa: on best rate}
The rate $\Opp(1 / T)$ obtained in statement (2) of Theorem \ref{thm: DIR for yields} cannot be improved. We shall now present an example where  $\plimT \pare{T (R^T_{s} - R^T_{t})}$ exists for all $s < t$, and is a \emph{nonzero} random variable. We shall use again the construction of Example \ref{exa: EMM construction}.

\smallskip

Consider the filtered probability space $\basis$, and let $\qprob = \prob$. Let also $W$ be a standard one-dimensional Brownian motion on the latter filtered probability space. The filtration $\filtration$ is assumed to be the one generated by $W$. Let $b \in \Real$. Define a \textsl{short-rate} process $r$ starting at some $r_0 \in \Real$, satisfying
\[
r_t = e^{-t} r_0 + (1 - e^{-t}) b - \sqrt{2} e^{-t} \int_0^t W_u e^{u} \ud u + \sqrt{2} W_t
\]
for all $t \in \Real_+$. In differential terms it is easy to see that $\ud r_t = (b - r_t) \ud t + \sqrt{2} \ud W_t$; this is a special case of the Vasicek model for the short rate --- see \cite{citeulike:1069961}. The parameters are chosen to simplify the formula \eqref{eq: yield vasicek} below for the yield. Let $B \dfn \exp \pare{\int_0^\cdot r_t \ud t}$ and define a market according to Example \ref{exa: EMM construction}. In this case it is well-known (see \cite{MR2255741}) that
\begin{equation} \label{eq: yield vasicek}
R_t^T = \frac{1 - e^{-(T - t)}}{T - t} \, r_t + \frac{\pare{1 - e^{-(T - t)}}^2}{2 (T - t)} + (b - 1) \pare{1 - \frac{1 - e^{-(T - t)}}{T - t}}.
\end{equation}
In particular, $\plimT R_t^T = b - 1$ holds for all $t \in \Real_+$, which implies that $\bbe_t \eqp \Omega$ for all $t \in \Real_+$. Using \eqref{eq: yield vasicek} once again we get $\plimT \pare{T R_t^T - T(b-1)} = r_t - b + 3/2$. Therefore, for $s < t$, $\plimT \pare{T (R_s^T - R_t^T) } \, = \, r_s - r_t$, which is a nontrivial Gaussian random variable.

\subsection{Market viability} \label{exa: determin counterexa}

As already discussed, asking for the existence of a strictly positive supermartingale deflator is a market viability condition that is weaker than the ones that have appeared previously in the literature. Here, we shall present an example of a market with \emph{deterministic} bond prices which admits a strictly positive supermartingale deflator, but where more classical viability assumptions fail.

The probability space we are working on is left intentionally unspecified, since it plays absolutely no role. For $0 \leq t \leq T$, define $P_t^T = \min \set{1, \, \exp(1 - (T - t)) }$. Since $R_t^T = 1 - 1 / (T - t)$ holds for $T > t + 1$, we obtain $\lim_{T \to \infty} R_t^T = 1$ for all $t \in \Real_+$. Therefore, $\Phi_t \eqp \Omega$ for all $t \in \Real_+$, and the results of Theorem \ref{thm: DIR for yields} and Theorem \ref{thm: DIR for forward} hold trivially. 

Let $Y$ be defined via $Y_t = \exp(-t)$ for $t \in \Real_+$. Then, $Y_t P_t^T = \min \set{\exp(-t), \, \exp(1 - T) }$ for $0 \leq t \leq T$, which means that $(Y_t P^T_t)_{t \in [0, T]}$ is a nonincreasing process, i.e., a supermartingale. It follows that a strictly positive supermartingale deflator exists in this market. It follows that there cannot exist any arbitrages of the first kind if we only consider long positions in the bonds. This follows from the existence of a strictly positive supermartingale deflator, in view of the general results in \cite{Kar08}. However, we shall shortly see that if we allow for short positions on short-term bonds, arbitrages appear.

Let $t \in \Real_+$. For any $T \geq t+2$, note that
\[
P_t^{T} = \exp(-T + t + 1) < P_{t}^{t+1} P_{t+1}^{T} = \exp(-T + t + 2).
\]
Consequently, there cannot exist a probability $\qprob_{t, t+1}$ such that $P_t^{T} \geq P_t^{t+1} \expec_{\qprob_{t, t+1}} \bra{P_{t+1}^T \such \F_{t+1}}$. Therefore, condition 2.10 of \cite{GolSch08}, which already is a weaker version of existence of an EMM, is not satisfied. Furthermore, consider the following investment strategy at time $t$: take a long position of $\exp (T - t - 1)$ units of a bond maturing at time $T \geq t+1$ and a short position in a single unit of a bond maturing at time $t+1$. The capital required for this position at time $t$ is $- \exp (T - t - 1) P_{t}^T - P_{t}^{t+1} = 1 - 1 = 0$. At time $t + 1$, the value of this position will be
\[
\exp (T - t - 1) P_{t+1}^{T} - P_{t+1}^{t+1} = \exp(1) - 1 > 0. 
\]
Therefore, there exists an arbitrage in the market according to Definition 2.29 from \cite{GolSch08} once we allow for short positions on short-term bonds. Observe that one does not even need the ``limiting'' procedure mentioned in \cite{GolSch08} in the definition of arbitrage.

\bibliographystyle{siam}
\bibliography{DIR}

\begin{thebibliography}{10}

\bibitem{MR2255741}
{\sc D.~Brigo and F.~Mercurio}, {\em Interest rate models---theory and practice
  (with smile, inflation and credit)}, Springer Finance, Springer-Verlag,
  Berlin, second~ed., 2006.

\bibitem{RePEc:uts:rpaper:198}
{\sc N.~Bruti-Liberati, C.~Nikitopoulos-Sklibosios, and E.~Platen}, {\em
  Real-world jump-diffusion term structure models}, Quant. Finance, 10 (2010),
  pp.~23--37.

\bibitem{MR1304434}
{\sc F.~Delbaen and W.~Schachermayer}, {\em A general version of the
  fundamental theorem of asset pricing}, Math. Ann., 300 (1994), pp.~463--520.

\bibitem{MR1779587}
{\sc F.~D{\"o}berlein, M.~Schweizer, and C.~Stricker}, {\em Implied savings
  accounts are unique}, Finance Stoch., 4 (2000), pp.~431--442.

\bibitem{DIR}
{\sc P.~H. Dybvig, J.~E.~J. Ingersoll, and S.~A. Ross}, {\em Long forward and
  zero-coupon rates can never fall}, Journal of Business, 69 (1996), pp.~1--25.

\bibitem{KFG}
{\sc N.~El~Karoui, A.~Frachot, and H.~Geman}, {\em On the behavior of long zero
  coupon rates in a no arbitrage framework}, Review of Derivatives Research, 1
  (1998), pp.~351--369.

\bibitem{GolSch08}
{\sc V.~Goldammer and U.~Schmock}, {\em Generalization of the
  {D}ybvig-{I}ngersoll-{R}oss theorem and asymptotic minimality}, tech. rep.,
  Vienna University of Technology, 2009.
\newblock Submitted; electronic preprint available at
  http://www.fam.tuwien.ac.at/$\sim$schmock/preprints/Dybvig-Ingersoll-Ross.pd%
f.

\bibitem{MR1219534}
{\sc S.~W. He, J.~G. Wang, and J.~A. Yan}, {\em Semimartingale theory and
  stochastic calculus}, Kexue Chubanshe (Science Press), Beijing, 1992.

\bibitem{MR1926241}
{\sc F.~Hubalek, I.~Klein, and J.~Teichmann}, {\em A general proof of the
  {D}ybvig-{I}ngersoll-{R}oss theorem: long forward rates can never fall},
  Math. Finance, 12 (2002), pp.~447--451.

\bibitem{Kar08}
{\sc C.~Kardaras}, {\em Generalized supermartingale deflators under limited
  information}.
\newblock Submitted. Draft available electronically at
  http://arxiv.org/abs/0904.2913, 2009.

\bibitem{Malkiel}
{\sc B.~G. Malkiel}, {\em The term structure of interest rates: Expectations
  and behavior patterns}, Princeton University Press, 1966.

\bibitem{RePEc:osu:osuewp:00-12}
{\sc J.~H. McCulloch}, {\em Long forward and zero-coupon rates indeed can never
  fall, but are indeterminate: A comment on {D}ybvig, {I}ngersoll and {R}oss},
  Working Paper 00-12, Ohio State University, Department of Economics, Sept.
  2000.

\bibitem{schulze}
{\sc K.~Schulze}, {\em Asymptotic maturity behavior of the term structure},
  bonn econ discussion papers, University of Bonn, Germany, June 2008.

\bibitem{citeulike:1069961}
{\sc O.~Vasicek}, {\em An equilibrium characterization of the term structure},
  Journal of Financial Economics, 5 (1977), pp.~177--188.

\bibitem{Yao1999327}
{\sc Y.~Yao}, {\em Term structure modeling and asymptotic long rate},
  Insurance: Mathematics and Economics, 25 (1999), pp.~327 -- 336.

\end{thebibliography}
\end{document}